\documentclass[a4paper,UKenglish,cleveref,hyperref,autoref]{lipics-v2021}

\title{Maximum Independent Set when excluding an induced minor: $K_1 + tK_2$ and $tC_3 \uplus C_4$}
\titlerunning{MIS when excluding an induced minor: $K_1 + tK_2$ and $tC_3 \uplus C_4$}

\author{\'{E}douard Bonnet$^*$}{$^*$Corresponding author: edouard.bonnet@ens-lyon.fr\\Univ Lyon, CNRS, ENS de Lyon, Université Claude Bernard Lyon 1, LIP UMR5668, France \and \url{http://perso.ens-lyon.fr/edouard.bonnet/}}{edouard.bonnet@ens-lyon.fr}{https://orcid.org/0000-0002-1653-5822}{}
\author{Julien Duron}{Univ Lyon, CNRS, ENS de Lyon, Université Claude Bernard Lyon 1, LIP UMR5668, France}{julien.duron@ens-lyon.fr}{}{}
\author{Colin Geniet}{Univ Lyon, CNRS, ENS de Lyon, Université Claude Bernard Lyon 1, LIP UMR5668, France}{colin.geniet@ens-lyon.fr}{https://orcid.org/0000-0003-4034-7634}{}
\author{Stéphan Thomassé}{Univ Lyon, CNRS, ENS de Lyon, Université Claude Bernard Lyon 1, LIP UMR5668, France}{stephan.thomasse@ens-lyon.fr}{}{}
\author{Alexandra Wesolek}{Simon Fraser University, Burnaby, BC, Canada}{agwesole@sfu.ca}{https://orcid.org/0000-0003-4841-5937}{Supported by the Vanier Canada Scholarship Program.}

\authorrunning{\'E. Bonnet, J. Duron, C. Geniet, S. Thomassé, A. Wesolek}

\Copyright{Édouard Bonnet, Julien Duron, Colin Geniet, Stéphan Thomassé, Alexandra Wesolek}



\category{}

\relatedversion{}

\supplement{}

\funding{This work was supported by the ANR projects TWIN-WIDTH (ANR-21-CE48-0014) and Digraphs (ANR-19-CE48-0013).}

\acknowledgements{%
  We would like to thank Dibyayan Chakraborty for useful discussions, and anonymous reviewers for their helpful comments, and in particular a nice simplification in the proof of~\cref{thm2}.
}

\nolinenumbers 

\hideLIPIcs  

\EventEditors{Inge Li G{\o}rtz, Martin Farach-Colton, Simon J. Puglisi, and Grzegorz Herman}
\EventNoEds{4}
\EventLongTitle{31st Annual European Symposium on Algorithms (ESA 2023)}
\EventShortTitle{ESA 2023}
\EventAcronym{ESA}
\EventYear{2023}
\EventDate{September 4--6, 2023}
\EventLocation{Amsterdam, the Netherlands}
\EventLogo{}
\SeriesVolume{274}
\ArticleNo{12}

\usepackage[utf8]{inputenc}  

\usepackage[T1]{fontenc}
\usepackage{lmodern}

\usepackage[colorinlistoftodos,bordercolor=orange,backgroundcolor=orange!20,linecolor=orange,textsize=normalsize]{todonotes}

\usepackage{amsmath}  
\usepackage{amssymb}     
\usepackage{bbm}
\usepackage{complexity}

\usepackage{booktabs}

\makeatletter
\newtheorem*{rep@theorem}{\rep@title}
\newcommand{\newreptheorem}[2]{%
\newenvironment{rep#1}[1]{%
 \def\rep@title{#2 \ref{##1}}%
 \begin{rep@theorem}}%
 {\end{rep@theorem}}}
\makeatother

\newreptheorem{theorem}{Theorem}
\newreptheorem{lemma}{Lemma}

\usepackage{xspace}
\usepackage{tikz}

\usepackage[ruled,vlined,linesnumbered]{algorithm2e}

\usetikzlibrary{fit}
\usetikzlibrary{arrows}
\usetikzlibrary{patterns}
\usetikzlibrary{calc}
\usetikzlibrary{shapes}
\usetikzlibrary{positioning}
\usetikzlibrary{math}
\usetikzlibrary{shapes.symbols}
\usetikzlibrary{decorations.pathreplacing,decorations.pathmorphing,decorations.shapes,calligraphy}
\usepackage[scr=boondox]{mathalfa}

\newcommand{\mis}{\textsc{Max Independent Set}\xspace}
\newcommand{\smis}{\textsc{MIS}\xspace}

\renewcommand{\geq}{\geqslant}
\renewcommand{\leq}{\leqslant}

\renewcommand{\ge}{\geq}

\newcommand\dist{\text{dist}}

\newtheorem{question}{Question}
\crefname{question}{Question}{Questions}
\crefname{table}{Table}{Tables}
\crefname{claim}{Claim}{Claims}
\crefname{reptheorem}{Theorem}{Theorems}

\begin{document}

\maketitle

\begin{abstract}
  Dallard, Milanič, and Štorgel [JCTB '24] ask if, for every class excluding a fixed planar graph $H$ as an induced minor, \textsc{Maximum Independent Set} can be solved in polynomial time, and show that this is indeed the case when $H$ is any planar complete bipartite graph, or the 5-vertex clique minus one edge, or minus two disjoint edges.
  A positive answer would constitute a far-reaching generalization of the state-of-the-art, when we currently do not know if a polynomial-time algorithm exists when $H$ is the 7-vertex path.
  Relaxing tractability to the existence of a quasipolynomial-time algorithm, we know substantially more.
  Indeed, quasipolynomial-time algorithms were recently obtained for the $t$-vertex cycle, $C_t$ [Gartland et al., STOC '21], and the disjoint union of $t$ triangles, $tC_3$ [Bonamy et al., SODA '23].

  We give, for every integer $t$, a polynomial-time algorithm running in $n^{O(t^5)}$ when $H$ is the friendship graph $K_1 + tK_2$ ($t$ disjoint edges plus a~vertex fully adjacent to them), and a quasipolynomial-time algorithm running in $n^{O(t^2 \log n)+f(t)}$, with $f$ a~single-exponential function, when $H$ is $tC_3 \uplus C_4$ (the disjoint union of $t$ triangles and a 4-vertex cycle).
  The former generalizes the algorithm readily obtained from Alekseev's structural result on graphs excluding $tK_2$ as an induced subgraph [Alekseev, DAM '07], while the latter extends Bonamy et~al.'s result.
\end{abstract}

\section{Introduction}\label{sec:intro}

The \mis (\smis for short) problem asks for a~largest \emph{independent set} of its input graph $G$, i.e., a~subset of pairwise non-adjacent vertices in $G$.
In its decision form, the input is a graph $G$ and an integer $k$, and the question is whether $G$ admits an independent set of size at least~$k$.

Besides the ubiquitous usefulness that such a fundamental problem has within combinatorial optimization, and notably in the areas of packing, scheduling, and coloring, \smis (or equivalently \textsc{Maximum Clique}, the same problem in the complement graph) has a~very wide range of applications, as evidenced, for instance, in map labeling~\cite{Verweij99}, coding theory~\cite{Butenko02}, spatial scheduling~\cite{Eddy21}, genetic analysis~\cite{Abraham14}, information retrieval~\cite{Augustson70}, macromolecular docking~\cite{Gardiner00}, and sociometry~\cite{Forsyth46} (also see Butenko's thesis~\cite{Butenko03}).
It is thus unfortunate that this problem is not only hard to solve but also very resistant to approximation.
Indeed, the decision version of \smis is NP-complete~\cite{GJ79} and W$[1]$-complete~\cite{Downey95}, while its optimization version cannot be approximated within a~factor $n^{1-\varepsilon}$ on $n$-vertex graphs, for any $\varepsilon > 0$, unless P$=$NP~\cite{Hastad96,Zuckerman07}.

In spite of this, theorists and practitioners have put a lot of effort into designing efficient algorithms for \smis. 
In parallel with generic approaches via integer programming, high-performance exact and heuristic \smis solvers have emerged in recent years, based on diverse methods such as kernelization and evolutionary approaches~\cite{Lamm17}, deep reinforcement learning~\cite{Ahn20}, graph neural networks~\cite{Pontoizeau21}, and dataless training (where backpropagation is applied to a loss function based instead on the input)~\cite{Alkhouri22}.
On the theory side, exact exponential algorithms have been developed for decades culminating in a~running time below $1.2^n n^{O(1)}$~\cite{Xiao17}.

Another approach is to try and exploit the structure that the input graphs may have.
Indeed, in all the aforementioned applications, inputs are not uniformly sampled over all $n$-vertex graphs: They instead bear some structural properties, and in some cases, might avoid some specific patterns.
Graph theory\footnote{We refer the reader to~\cref{sec:prelim} for the relevant background in graph theory.} offers two main notions of \emph{patterns} or containment: the natural and straightforward \emph{subgraphs} (obtained by removing vertices and edges), and the deeper \emph{minors} (further allowing to contract edges).
Both notions come with an induced variant, when edge removals are disallowed, bringing the number of containment types to four.
It is then sensible to determine the patterns $H$ whose absence makes \smis (more) tractable.
It turns out that this question is completely settled for subgraphs and minors.

For the subgraph containment, the argument is the following.
By the grid minor theorem~\cite{Robertson86}, the class of graphs excluding $H$ as a subgraph has bounded treewidth if (and only if) all the connected components of $H$ are paths and subdivided \emph{claws} (i.e., stars with three leaves); thus \smis can be solved in polynomial time in this class, for instance by Courcelle's theorem~\cite{Courcelle90}.
If instead $H$ has a~connected component that is neither a~path nor a~subdivided claw, \smis remains NP-complete since such a~class either contains all subcubic graphs or the $2|V(H)|$-subdivision of every graph, two families of graphs on which \smis is known to be NP-complete~\cite{Alekseev82,AlKa00,Poljak74}.

For minors, the dichotomy hinges on whether $H$ is planar.
Indeed, if $H$ is planar, then the class of graphs excluding $H$ as a minor has bounded treewidth (again, mainly by the grid minor theorem), and \smis can be solved efficiently.
If $H$ is non-planar, then the $H$-minor-free graphs include all planar graphs for which \smis is known to be NP-complete~\cite{Garey76}.

The question is more intriguing for the induced containments, and the \emph{induced subgraph} case has received a lot of attention.
While it has long been known that if $H$ is \emph{not} the disjoint union of paths and subdivided claws, \smis remains NP-complete on graphs without $H$ as an induced subgraph~\cite{Alekseev82,Poljak74}, it has been conjectured that \smis is otherwise polynomial-time solvable.
This has been proven when $H$ is the 6-vertex path \cite{GrzesikKPP22}, a claw with exactly one edge subdivided~\cite{Alekseev04,Lozin08}, or any disjoint union of claws~\cite{BrandstadtM18a}.
The latter result extends a~polynomial-time algorithm (essentially) due to Alekseev when $H$ is any disjoint union of edges~\cite{Alekseev07}.
The author indeed proves that the total number of maximal independent sets is polynomially bounded.
One can then enumerate all the maximal independent sets in polynomial time (following Alekseev's proof, or using the~generic output-sensitive algorithm of Tsukiyama et al.~\cite{Tsukiyama77}), and thus find a~maximum independent set.

While we currently do not know of a polynomial-time algorithm when $H$ is the 7-vertex path, Gartland and Lokshtanov~\cite{Gartland20} have obtained a quasipolynomial-time algorithm when $H$ is $P_t$, the $t$-vertex path, for any positive integer $t$; also see~\cite{PilipczukPR21}. 
Supporting the existence of a~polynomial-time algorithm when $H$ is $S_{i,j,k}$, the claw whose three edges are subdivided $i-1$, $j-1$, and $k-1$ times, respectively, a~quasipolynomial-time algorithm \cite{Gartland24}, and polynomial-time algorithms among bounded-degree graphs~\cite{Abrishami22}, and later, among graphs excluding a~fixed complete bipartite graph as a~subgraph \cite{Abrishami24} have been found.
The parameterized complexity of \smis when excluding a fixed induced subgraph has been studied~\cite{BonnetBCTW18,BonnetBTW19,Dabrowski12}, but the mere statement of which $H$ make the problem fixed-parameter tractable (and which ones keep it W$[1]$-complete) is unclear~\cite{BonnetBTW19}.

We eventually arrive at the \emph{induced minor} containment, the topic of the current paper.
As for minors, the class of all graphs excluding a~non-planar graph $H$ as an induced minor contains all planar graphs; hence \smis remains NP-complete in such a class.
However we do not know of a~\emph{planar} graph $H$ for which \smis remains NP-complete on $H$-induced-minor-free graphs.
This has led Dallard, Milanič, and Štorgel~\cite{Dallard-3} to ask if such classes exist:
\begin{question}\label{q:poly}
  Is it true that for every planar graph $H$, \mis can be solved in polynomial time in the class of graphs excluding $H$ as an induced minor?
\end{question}
A~first observation is that avoiding $H$ as an induced minor implies avoiding it as an induced subgraph.
Thus \cref{q:poly} is settled for $P_6$, $S_{1,1,2}$, and $tS_{1,1,1}$ (where $tG$ denotes the disjoint union of $t$ copies of $G$).
The same authors~\cite{Dallard-3} further obtain a~polynomial-time algorithm when $H$ is $K_5^-$ (the 5-vertex clique minus an edge), $K_{2,t}$ (the complete bipartite graph with 2 vertices fully adjacent to $t$ vertices), and $W_4=K_1+C_4$ (a 4-vertex cycle $C_4$ with a fifth vertex fully adjacent to the cycle).
All three cases were shown by bounding the so-called \emph{tree-independence number} (i.e., treewidth where \emph{bag size} is replaced by \emph{independence number of the subgraph induced by the bag})~\cite{Dallard-3}, in which case a~polynomial-time algorithm can be derived for \smis using the corresponding tree-decompositions~\cite{Dallard-2}.
They also show that this is as far as this sole technique can go: $H$-induced-minor-free graphs have bounded \emph{tree-independence number} if and only if $H$ is edgeless or an induced minor of $K_5^-$, $K_{2,t}$, or $W_4$~\cite{Dallard-3}.
The framework of \emph{potential maximal cliques}~\cite{Bouchitte01} and the \emph{container method} have led to a polynomial-time algorithm when $H=C_5$~\cite{Abrishami21,Chudnovsky20}.

\cref{q:poly} is a~beautiful question and, if true, a very difficult one.
Indeed, $H=P_7$, the 7-vertex path, is a~very simple planar graph for which we currently do not know such a~polynomial-time algorithm.
A~natural relaxation of~\cref{q:poly} is to only request a~quasipolynomial-time algorithm:
\begin{question}\label{q:qpoly}
  Is it true that for every planar graph $H$, \mis can be solved in quasipolynomial time in the class of graphs excluding $H$ as an induced minor?
\end{question}

We know somewhat more about~\cref{q:qpoly}.
There is a~quasipolynomial-time algorithm for \smis in $C_t$-induced-minor-free graphs~\cite{Gartland21}, building upon the $H=P_t$ case.
Recently, Bonamy et al.~\cite{Bonamy23} presented a~quasipolynomial-time algorithm when $H$ is $tC_3$, i.e., the disjoint union of $t$ triangles.
(When $K_{1,t}$, the star with $t$ leaves, is further excluded as an induced subgraph, a~polynomial-time algorithm has been obtained~\cite{Ahn25}.)
See~\cref{tbl:mis-complexity} for a~summary of the introduction.

\begin{table}
  \begin{tabular}{lcccc}
    \toprule
    $H$ excluded as     & subgraph                    & minor          & induced subgraph               & \textbf{induced minor}     \\
    \midrule
 in P                  & $tS_{t,t,t}$                  & planar         &
 $P_6$, $S_{1,1,2}$, $tS_{1,1,1}$  & $K_5^-\!$, $K_{2,t}$, $W_4$, $C_5$, \framebox{$K_1 \! + \! tK_2$} \\
 known in QP      & $-$                 &  $-$   &  $tS_{t,t,t}$                   & $C_t$ ($t\geqslant 6$), \framebox{$tC_3 \uplus C_4$}         \\
 NP-c                  & $\neg tS_{t,t,t}$  &  non-planar     &  $\neg tS_{t,t,t}$             &     non-planar                       \\
 open P / NP-c        & $-$                &  $-$    &  $P_7$, \dots                   &    $P_7$, $C_6$, \dots                        \\
 open QP / NP-c       & $-$                &  $-$    &  $-$    &    $C_4 \uplus C_4$, \dots              \\
    \bottomrule
  \end{tabular}
  \caption{The complexity of \mis when $H$ is excluded as one of the four main types of patterns.
    ``$\neg tS_{t,t,t}$'' means that $H$ is not a subgraph of $tS_{t,t,t}$ for any $t$.
    Our results are framed.}
   \label{tbl:mis-complexity}
\end{table}

Expecting an affirmative solution to \cref{q:poly} or \cref{q:qpoly} may seem optimistic.
However, regarding precise running time, we do know that for every planar $H$, \smis is probably not as difficult in $H$-induced-minor-free graphs as it is in general graphs.
Indeed, Korhonen~\cite{Korhonen22} describes a~$2^{O(n/\log^c n)}$-time algorithm, for some constant $c>0$, to solve \smis on $n$-vertex graphs excluding a~fixed planar graph $H$ as an induced minor.
Assuming the Exponential-Time Hypothesis~\cite{ImpagliazzoP01}, such a running time is impossible in general graphs~\cite{ImpagliazzoPZ01}.

\medskip

\textbf{Our results.}
We make some progress regarding \cref{q:poly,q:qpoly}.
Our first contribution is, for every positive integer $t$, a polynomial-time algorithm when $H$ is the \emph{friendship graph} $K_1+tK_2$ (also called \emph{Dutch windmill graph} or \emph{fan}), i.e., $t$ independent edges universally linked to a~$(2t+1)$-st vertex:
\begin{theorem}\label{thm1}
  For every positive integer $t$, \mis can be solved in polynomial-time $n^{O(t^5)}$ in $n$-vertex $K_1+tK_2$-induced-minor-free graphs.
\end{theorem}
This extends Alekseev's result~\cite{Alekseev07} for graphs excluding $tK_2$ as an induced subgraph, or equivalently, as an induced minor.
We indeed use this result to first derive a~polynomial-time algorithm in subgraphs of $K_1+tK_2$-induced-minor-free graphs $G$ induced by vertices from a~bounded number of breadth-first search (BFS) layers of $G$.

We then consider the connected components of our input graph $G$ when deprived of a~subset~$X$ of vertices inducing $tK_2$ and, subject to that property, maximizing the order of the largest connected component in $G-X$. 
We show that, due to this careful selection of~$X$, every component $C$ of $G-X$ admits an efficiently constructible path-decomposition $\mathcal P$ with \emph{bounded adhesion} (i.e., any two distinct bags have a~bounded intersection), each bag of which is contained in a~bounded number of consecutive BFS layers of $C$.
Hence \smis can be solved efficiently within a~bag, by our opening step (see previous paragraph).
This part is quite technical, but mostly to justify the existence of $\mathcal P$.
The algorithm itself remains simple.

\cref{thm1} is then obtained by exhaustively finding $X$ and guessing its intersection~$X'$ with a~maximum independent set of $G$, and performing dynamic programming on the connected components of $G-X$, deprived of $N(X')$.
The dynamic-programming table is filled via the efficient algorithm when handling an induced subgraph contained in few BFS layers.

\medskip

Our second contribution is a quasipolynomial-time algorithm when $H$ is $tC_3 \uplus C_4$, the disjoint union of $t$ triangles and a~4-vertex cycle:
\begin{theorem}\label{thm2}
  For every positive integer $t$, \mis can be solved in quasipolynomial-time $n^{O(t^2 \log n)+f(t)}$ (where $f$ is single-exponential) in $tC_3 \uplus C_4$-induced-minor-free graphs.
\end{theorem}
We first perform a~quasipolynomial branching rule to get rid of holes of length at most 6 (i.e., induced cycles of length $4$, $5$, or $6$). 
We then assume that the input graph $G$ is not $(t+2)C_3$-induced-minor-free, for otherwise we conclude with Bonamy et al.'s algorithm~\cite{Bonamy23}.
Thus $G$, being $tC_3 \uplus C_4$-induced-minor-free, has to admit $(t+2)C_3$ as an induced subgraph, i.e., a collection $T_1, \ldots, T_{t+2}$ of $t+2$ pairwise vertex-disjoint and non-adjacent triangles.
We define $S_{i,j}$, minimally separating $T_i$ and $T_j$ in the graph $G$ deprived of the neighborhoods of the other triangles $T_k$ (with $k \neq i, j$).

We show that each $S_{i,j}$ induces a clique.
So does every intersection $N_{i,j}$ of the neighborhood of two distinct triangles $T_i, T_j$ of the collection (this is where getting rid of the holes of length at most~6 comes into play).
We can therefore exhaustively guess the intersection of a~maximum independent set with the union of the sets $S_{i,j}$ and $N_{i,j}$ (for every $i < j \in [t+2]$).
We finally observe that $G'=G-\bigcup_{i \neq j \in [t+2]}(S_{i,j} \cup N_{i,j})$ is chordal, since the presence of a~hole $H$ in $G'$ would imply the existence in $G$ of $t$ independent triangles in the non-neighborhood of $H$, a~contradiction to the $tC_3 \uplus C_4$-induced-minor-freeness of $G$.
We thus conclude by using a~classic algorithm for \smis in chordal graphs~\cite{Gavril74,Rose76}.

\medskip

In~\cref{sec:prelim} we introduce the relevant graph-theoretic background.
In~\cref{sec:p} we prove~\cref{thm1}, and in~\cref{sec:qp} we prove \cref{thm2}.

\section{Preliminaries}\label{sec:prelim}

If $i \leqslant j$ are two integers, we denote by $[i,j]$ the set of integers $\{i,i+1,\ldots,j-1,j\}$, and by~$[i]$, the set $[1,i]$.
We denote by $V(G)$ and $E(G)$ the set of vertices and edges of a graph~$G$, respectively.
We denote by $G_1 \simeq G_2$ the fact that the two graphs $G_1$ and $G_2$ are \emph{isomorphic}, i.e., equal up to renaming their vertex sets.
For $S \subseteq V(G)$, the \emph{subgraph of $G$ induced by $S$}, denoted $G[S]$, is obtained by removing from $G$ all the vertices that are not in $S$ (together with their incident edges).
Then $G-S$ is a shorthand for $G[V(G)\setminus S]$.
A~graph $H$ is an \emph{induced subgraph} of $G$ if there is an $S \subseteq V(G)$ such that $G[S] \simeq H$.

For $G$ a graph and a set $X \subseteq V(G)$, $E_G(X)$ (or simply $E(X)$) is a shorthand for $E(G[X])$.
For $G$ a graph and $X,Y \subseteq V(G)$ two disjoint sets, $E_G(X,Y)$ denotes the set of edges of $E(G)$ with one endpoint in $X$ and the other endpoint in $Y$.
We denote by $N_G(v)$ and $N_G[v]$, the open, respectively closed, neighborhood of $v$ in $G$.
For $S \subseteq V(G)$, we set $N_G(S) := \bigcup_{v \in S}N_G(v) \setminus S$ and $N_G[S] := N_G(S) \cup S$.
We may omit the subscript if $G$ is clear from the context.
A~\emph{connected component} is a~maximal connected induced subgraph.

Two cycles $C, C'$ are said to be \emph{independent} if they are vertex-disjoint and there is no edge between $C$ and $C'$.
A collection of cycles is \emph{independent} if they are pairwise independent.
Two vertex subsets $X, Y \subseteq V(G)$ \emph{touch} if $X \cap Y \neq \emptyset$ or there is an edge $uv \in E(G)$ with $u \in X$ and $v \in Y$.
Then two (or more) cycles are independent if and only if they do \emph{not} touch.
We say that $X, Y \subseteq V(G)$ \emph{touch in $Z$} if $X \cap Y \cap Z \neq \emptyset$ or there is an edge $uv \in E(G)$ with $u \in X \cap Z$ and $v \in Y \cap Z$, or equivalently, if $X \cap Z$ and $Y \cap Z$ touch.

A graph $H$ is an \emph{induced minor} of a graph $G$ if $H$ can be obtained from $G$ by a sequence of vertex deletions and edge contractions.
A~\emph{minor} is the same but also allows edge deletions.
Equivalently an induced minor $H$---with vertex set, say, $\{v_1, \ldots, v_{|V(H)|}\}$---of $G$ can be defined as a vertex-partition $B_1, \ldots, B_{|V(H)|}$ of an induced subgraph of $G$, such that every $G[B_i]$ is connected and $v_iv_j \in E(H)$ if and only if $E_G(B_i,B_j) \neq \emptyset$ (i.e., when the disjoint sets $B_i$ and $B_j$ touch).
Observe indeed that contracting each $B_i$ into a single vertex (which is possible since each $B_i$ induces a~connected subgraph) results in $H$.
A graph $G$ (resp.~a graph class) is said to be \emph{$H$-induced-minor-free} if $H$ is not an induced minor of $G$ (resp.~no graph of the class admits $H$ as an induced minor).

We denote by $C_\ell$ the $\ell$-vertex cycle, and by $K_\ell$, the $\ell$-vertex complete graph.
A~\emph{hole} is an induced cycle of length at least four.
A~graph is~\emph{chordal} if it has no hole.
For two disjoint sets $X, Y \subseteq V(G)$ in a graph $G$, an \emph{$(X,Y)$-separator} is a (possibly empty) set $S \subseteq V(G) \setminus (X \cup Y)$ such that there is no path between $X$ and $Y$ in $G-S$.
An $(X,Y)$-separator is \emph{minimal} if no proper subset of it is itself an $(X,Y)$-separator.

The \emph{disjoint union} $G_1 \uplus G_2$ of two graphs $G_1, G_2$ has vertex set $V(G_1) \uplus V(G_2)$ and edge set $E(G_1) \uplus E(G_2)$, where $V(G_1) \uplus V(G_2)$ presupposes that the vertex sets of $G_1$ and $G_2$ are disjoint.
If $t \geqslant 2$ is an integer and $G$ a graph, $tG$ is the graph $G \uplus (t-1)G$, and $1G$ is simply~$G$.
The~\emph{join} $G_1 + G_2$ of two graphs $G_1, G_2$ has vertex set $V(G_1) \uplus V(G_2)$ and edge set $E(G_1) \uplus E(G_2) \uplus \{uv~:~u \in V(G_1), v \in V(G_2)\}$.
In other words, the join of $G_1$ and $G_2$ is obtained from their disjoint union by adding all possible edges between $G_1$ and $G_2$.

A~\emph{breadth-first search} (BFS) \emph{layering} in $G$ from a~vertex $v \in V(G)$ (or from a~connected set $S \subseteq V(G)$) is a partition of the remaining vertices into $L_1, L_2, \ldots$ such that every vertex of $L_i$ is at distance exactly $i$ from $v$ (or from $S$).
Such an $L_i$ is called a \emph{BFS layer} of $G$ (from $v$, or from $S$).
Note that there cannot be an edge in $G$ between $L_i$ and $L_j$ if $|i-j|>1$.

A~\emph{path-decomposition} of a graph $G$ is a list of vertex subsets $\mathcal P=(B_1, \ldots, B_h)$ such that
\begin{itemize}
\item $\bigcup_{1 \leq i \leq h} B_i = V(G)$,
\item for every $e \in E(G)$, there is some $B_i$ that contains both endpoints of $e$, and
\item whenever $v \in B_i \cap B_j$ with $i<j$, $v$ is also in all $B_k$ with $i < k < j$.
\end{itemize}
The sets $B_i$ (for $i \in [h]$) are called the \emph{bags} of $\mathcal P$, and the sets $B_i \cap B_{i+1}$ (for $i \in [h-1]$) the \emph{adhesions} of $\mathcal P$.
Path-decomposition $\mathcal P$ has maximum adhesion~$p$ if all of its adhesions have size at most~$p$.
Note that the adhesion $B_i \cap B_{i+1}$, if disjoint from $B_1 \cup B_h$, is a vertex cutset disconnecting $B_1$ from $B_h$.

\section{Polynomial algorithm in $K_1 + tK_2$-induced-minor-free graphs}\label{sec:p}

We first show how to solve \mis in $K_1 + tK_2$-induced-minor-free graphs of bounded diameter.
More generally, we show the following.

\begin{lemma}\label{lem:bounded-diameter}
  Let $t, h$ be fixed non-negative integers.
  Let $G$ be a $K_1 + tK_2$-induced-minor-free $n$-vertex graph, and $L_0 \subseteq V(G)$ such that $G[L_0]$ is connected.
  Let $L_i$, for any $i \in [h]$, be the subset of vertices of $G$ at distance exactly $i$ from $L_0$.
  Then, given as input $G$, $L_0$, and $S \subseteq \bigcup_{1 \leqslant i \leqslant h} L_i$, a~maximum independent set of $H := G[S]$ can be computed in polynomial time $n^{(2t-1)h+O(1)}$.
\end{lemma}
\begin{proof}
  For every $j \in [h]$, $G[L_j]$ has no $tK_2$ induced subgraph (or equivalently, induced minor).
  Indeed, $G[\bigcup_{0 \leqslant i \leqslant j-1} L_i]$ is a connected graph, hence $\bigcup_{0 \leqslant i \leqslant j-1} L_i$ can be contracted to a single vertex, and every vertex in $L_j$ has at least one neighbor in $L_{j-1}$.
  Therefore a~$tK_2$ induced subgraph in $G[L_j]$ would contradict the $K_1 + tK_2$-induced-minor-freeness of~$G$.

  Fix an arbitrary $S \subseteq \bigcup_{1 \leqslant i \leqslant h} L_i$, and consider the induced subgraph~$H := G[S]$.
  In particular $H[L_j \cap S]$ has also no $tK_2$ induced subgraph, for every $j \in [h]$.
  Hence, by a~classical result of Alekseev~\cite{Alekseev07}, $H[L_j \cap S]$ has at most $n^{2t-1}$ maximal independent sets, which can be listed in time $n^{2t+O(1)}$~\cite{Tsukiyama77}.

  We thus exhaustively list every $h$-tuple $(I_1, \ldots, I_h)$ where, for every $j \in [h]$, $I_j$ is a~maximal independent set of $H[L_j \cap S]$, in time $n^{(2t-1)h+O(1)}$.
  Note that if there is an edge in $H$ between $L_i$ and $L_j$, then $|i-j| \leqslant 1$.
  As each $I_j$ (for $j \in [h]$) is an independent set, $H'=H[\bigcup_{j \in [h]} I_j]$ is a bipartite graph as witnessed by the bipartition $(I_1 \cup I_3 \cup \ldots, I_2 \cup I_4 \cup \ldots)$.
  A~maximum independent set $I$ can thus be computed in polynomial time in $H'$.
  Indeed, by the K\H{o}nig-Egerv\'ary theorem~\cite{Konig31}, finding a~maximum independent set in a~bipartite graph boils down to finding a~maximum matching, which can be done in polynomial time (and now, even in almost linear time~\cite{Chen22}) by solving a~maximum flow problem.
  We output the largest independent set $I$ found among every run.
  
  The correctness of the algorithm is based on the observation that a~maximum independent set $I^\star$ of $H$ intersects every $L_i$ (for $i \in [h]$) in an independent set $J_i$, which by definition is contained in a maximal independent set $I_i$ of $H[L_i \cap S]$.
  In the run when every maximal independent set $I_i$ is a~superset of $J_i$, we obtain an independent set with cardinality equal to that of~$I^\star$.
\end{proof}

We say that $G$ is \emph{reduced} if it does not contain degree-1 vertices or degree-2 vertices with adjacent neighbors (i.e., part of a~triangle).
If $e=uv$ is an edge of~$G$, let $G\setminus e$ (resp. $S \setminus e$, $S \cup e$, for some $S \subseteq V(G)$) be the induced subgraph $G[V(G)-\{u,v\}]=G-\{u,v\}$ (resp.~the sets $S \setminus \{u,v\}$, $S \cup \{u,v\}$).
More generally, for a collection $e_1=u_1v_1, \ldots, e_k=u_kv_k$ of edges of $G$, we denote by $G \setminus \{e_1, \ldots, e_k\}$ the induced subgraph $G[V(G)-\{u_1, v_1, \ldots, u_k, v_k\}]=G-\{u_1, v_1, \ldots, u_k, v_k\}$.

\begin{lemma}\label{lem:edgesinC}
  Let $G$ be a reduced connected $K_1 + tK_2$-induced-minor-free graph containing $tK_2$ as an induced subgraph.
  Let $X \subseteq V(G)$ maximize the order of~a largest component of $G' := G - X$, among those sets $X$ such that $G[X] \simeq tK_2$.
  Then for any $e \in E(G'-N_G(X))$ contained in a connected component $C$ of $G'$,
  \begin{enumerate}
  \item $C \setminus e$ is disconnected, and
  \item each connected component of $C \setminus e$ contains a vertex in $N_G(X)$.
  \end{enumerate}
\end{lemma}
\begin{proof}
  As $G$ is reduced, every vertex of $X$ has degree at least~2 (in $G$).
  Thus $G'$ cannot be connected, for otherwise, contracting in $G$ the set $V(G')$ to a single vertex would form a~$K_1+tK_2$ induced minor.
  We thus know that $G'$ has at least two connected components.

  Let $C'$ be a largest connected component of $G'$.
  Since $G$ is connected, there exists a~shortest path $P$ in $G$ from $V(C')$ to $V(G')\setminus V(C')$.
  Say that $P$ ends in a connected component $C \neq C'$ of $G'$.
  Path $P$ has to have some internal vertices in~$X$, but since $G[X] \simeq tK_2$, it follows that there is an edge $e^*$ in $E(X)$ intersecting both $N_G(V(C))$ and $N_G(V(C'))$.
  
  For every edge $e \in E(G'-N_G(X))$ in component $C$ (which is possibly equal to $C'$), $C \setminus e$ is disconnected.
  Indeed, for the sake of contradiction, suppose that $C \setminus e$ is connected, and consider $X' := (X \setminus e^*) \cup e$.
  By assumption, $G[X'] \simeq tK_2$.
  Furthermore, the connected component of $G-X'$ containing $e^*$ is strictly larger than $C'$, as it contains $(V(C') \setminus e) \cup e^*$ and intersects $V(C) \setminus e$, which are two disjoint sets.
  This contradicts the maximality of $X$, and establishes the first item.

  We now prove the second item, also by contradiction.
  Suppose that there is a connected component $D$ of $C \setminus e$ that does not contain a~vertex in $N_G(X)$.
  We will reach a~contradiction by showing that $D$ contains an edge $e'$ not intersecting $N_G(X)$, and such that $D \setminus e'$ is connected (and conclude in light of the previous paragraph).
  
  Let $L_i \subseteq V(D)$ be the \emph{$i$-th neighborhood of $e$} in $D$, i.e., the vertices at distance $i$ of one endpoint of $e$, and at least $i$ of the other endpoint.
  We consider the \emph{last layer} $L_k$, i.e., such that $L_k \neq \emptyset$ and $L_{k+1} = \emptyset$.
  If $L_k$ contains an edge $e'$, then removing the endpoints of this edge does not disconnect $D$ (and hence $C$) since each vertex in $L_k$ has a neighbor in $L_{k-1}$ (with the convention that $L_0$ consists of the endpoints of $e$) and $G[L_0 \cup L_1 \cup \dots \cup L_{k-1}]$ is connected.

  If $L_k$ does not contain any edge, then $k \geqslant 2$.
  Otherwise (if $k=1$), then $D$ has a~single vertex, say, $w$, and $w$ has no neighbors in~$X$.
  Hence, in $G$, $w$ is only adjacent to one or the two endpoints of~$e$, contradicting that $G$ is reduced.

  Furthermore, each vertex in $L_k$ has two neighbors in $L_{k-1}$.
  This is because $G$ has minimum degree at~least~2 (as $G$ is reduced), and by assumption that no vertex of $D$ has a~neighbor in $X$.
  Hence, removing the endpoints of any edge $e'$ incident to a vertex in $L_k$ does not disconnect $D$, nor~$C$, since each vertex in $L_k$ has at least one neighbor in $L_{k-1}$ which is not an endpoint of~$e'$.
  In either case, $e'$~is an edge of $C - N_G(X)$ that does not disconnect $C$, which we showed is not possible. 
\end{proof}

We now prove the main technical result of the section.

\begin{proposition}\label{prop:path-dec}
  Let $G$ be a reduced connected $n$-vertex $K_1 + tK_2$-induced-minor-free graph containing $tK_2$ as an induced subgraph.
  Let $X \subseteq V(G)$ maximize the order of~a largest component of $G' := G - X$, among those sets $X$ such that $G[X] \simeq tK_2$.
  Then for every connected component $C$ of $G'$, a path-decomposition~$\mathcal P$ of $C$ such that 
\begin{itemize}
\item every bag of $\mathcal P$ is contained in $O(t^4)$ consecutive BFS layers of $C$, and
\item every adhesion of $\mathcal P$ is of size at most $2t^2$,
\end{itemize}
can be computed in time $n^{O(1)}$.
\end{proposition}

\begin{proof}
  Let $v$ be a vertex in $N_G(X) \cap V(C)$ and, for any positive integer $s$, let $L_s$ be the set of vertices at distance $s$ from $v$ in $C$.
  Let $q$ be the largest distance between $v$ and a~vertex of~$C$.
  We set $f(t) := (t^2+1)(6t^2+2)=O(t^4)$.
  We will show that, for every $s \in [q-f(t)]$, there is a~vertex cutset of size at most $2t^2$ separating $L_s$ from $L_{s+f(t)-1}$, and use that fact for $s=1, f(t)+1, 2f(t)+1, \ldots$ to build the path-decomposition $\mathcal P$.
  
We show that any sufficiently long induced path (such as a~shortest path from $L_s$ to $L_{s+f(t)-1}$) has some edges with both endpoints in $V(C) \setminus N_G(X)$.
\begin{claim}
  \label{cl:path_neighbors}
  Any induced path $P$ in $C$ contains less than $3t^2$ vertices in $N_G(X)$. 
\end{claim}
\begin{claimproof}
  If there are $3t^2$ vertices on $P$ with a neighbor in $X$, then there are at least $3t$ vertices $w_1, \ldots, w_{3t}$ on $P$ that are neighbors of a~fixed edge $e \in E(X)$.
  For each $i \in [t]$, contract every edge of $P$ between $w_{3i-2}$ and $w_{3i-1}$ but one, say $e_i$.
  Contract $e$, and call~$z$ the resulting vertex.
  The vertex $z$ and the $t$ edges $e_i$ contradict the fact that $G$ is $K_1+tK_2$-induced-minor-free. 
\end{claimproof}
For an edge $e \in E(C)$ we denote by $\dist(e,v)$ the length of a~shortest path in $C$ from an endpoint of~$e$ to $v$.
We build a collection of paths $Q_i$ of $C$, and edges $e_i \in E(Q_i)$, for $i=1, 2, \ldots$, while they are well-defined, in the following way. 

Let $s \in [q-f(t)]$ and $Q_1$ be a shortest path from $L_s$ to $L_{s+f(t)-1}$ in $C$.
Let $e_1 \in E(Q_1)$ minimize $e \mapsto \dist(e,v)$, among those edges of $Q_1$ with both endpoints in $V(C) \setminus N_G(X)$.
By~\cref{cl:path_neighbors} there are less than $6t^2$ edges on $Q_1$ with an endpoint in $N_G(X)$, hence $\dist(e_1,v) \leq s + 6t^2$.
We denote by $Q'_1$ the maximal subpath of $Q_1$ starting in $L_s$ and not containing an endpoint of~$e_1$ (that is, stopping just before reaching an endpoint of~$e_1$). Note that $Q_1$ is possibly empty. For the next iteration, we work in $C \setminus e_1$ (recall that this stands for $C$ deprived of the two endpoints of $e_1$).

We now describe in general the $i$-th iteration for $i \ge 2$.
Let $Q_i$ be a~shortest path from $L_s$ to $L_{s+f(t)-1}$ in $C \setminus \{e_1,e_2,\dots,e_{i-1}\}$.
Let $e_i$ be the first edge of $Q_i$ (when starting from $L_s$) such that $\dist(e_i,v) \geq \dist(e_{i-1},v)+2$ and $e_i$ has no endpoint in $N_G(X)$.
Note that by~\cref{cl:path_neighbors}, $\dist(e_i,v) \leq \dist(e_{i-1},v)+2+6t^2$.
Let $Q'_i$ be the maximal subpath of $Q_i$ starting in $L_s$ and not containing an endpoint of~$e_i$.
See~\cref{fig:paths-Pi-Qi}.

 \begin{figure}[ht]
    \centering
    \begin{tikzpicture}[vertex/.style={draw,circle,inner sep=-0.08cm}, vertexf/.style={fill,circle,inner sep=-0.07cm},
        vertexi/.style={circle,inner sep=-0.08cm},vertexb/.style={draw,circle,fill=white,inner sep=-0.08cm}]
      \def\t{6}
      \def\s{1}
      \def\ss{0.5}
      \def\dg{green!50!black}
      \foreach \i in {1,...,\t}{
        \node[vertexi] (a\i) at (\s * \i,2) {} ;
        \node[vertexi] (b\i) at (\s * \i,2 + 0.6 * \s) {} ;
      }
      \foreach \i/\j/\k in {{(a1)(b1)(a\t)(b\t)}/X/0.08}{
        \node[draw,rectangle,thick,rounded corners,fit=\i,inner sep=\k cm,fill,fill opacity=0.1] (\j) {} ;
      }
       \foreach \i in {1,...,\t}{
        \node[vertexb] at (\s * \i,2) {} ;
        \node[vertexb] at (\s * \i,2 + 0.6 * \s) {} ;
        \draw (a\i) -- (b\i) ; 
       }
     
      \foreach \i/\j in {1/4,2/4,3/5,4/4,5/5,6/5,7/4,8/5,9/5,11/4,12/3}{
        \foreach \k in {1,...,\j}{
          \node[vertexi] (v\i\k) at (\i,1-\ss * \k + \ss) {} ; 
        }
        \foreach \a/\b/\c in {{(v\i1)(v\i\j)}/{L\i}/0.08}{
        \node[preaction={fill},draw,rectangle,thick,rounded corners,fit=\a,inner sep=\c cm,fill opacity=0.15] (\b) {} ;
      }
      }
       \foreach \i/\j in {1/4,2/4,3/5,4/4,5/5,6/5,7/4,8/5,9/5,11/4,12/3}{
        \foreach \k in {1,...,\j}{
          \node[vertexb] at (\i,1-\ss * \k + \ss) {} ; 
        }
      }
      
      \foreach \i/\j/\k/\f in {0/1/v/f}{
        \node[vertex\f] (\k) at (\i,\j) {} ;
      }
      \foreach \i/\j in {1/2,2/1,2/2,4/4,6/1,7/4,9/1,9/3,11/2}{
        \node[vertexf] at (\i,1-\ss * \j + \ss) {} ;
      }
      
      \foreach \i/\j/\k in {0/0.75/{$v$},0.4/2.3/$X$,10/0/$\ldots$}{
        \node at (\i,\j) {\k} ;
      }

      \foreach \i/\j/\k in {a2/v/0,b1/v/0,v/v11/0,v/v12/0,v/v13/0,v/v14/0}{
        \draw (\i) to [bend left=\k] (\j) ;
      }
      \foreach \i/\j/\k in {v11/v22/blue,v22/v34/blue,v34/v44/blue, v13/v22/\dg,v22/v32/\dg,v32/v43/\dg,v43/v44/\dg,v44/v55/\dg,v55/v65/\dg}{
        \draw[\k,line width=0.1cm] (\i) -- (\j) ;
      }
      \foreach \i/\j/\k in {v12/v23/red,v35/v42/red,v42/v52/red,v52/v61/red,v61/v72/red,v72/v84/red,v84/v92/red, v112/v123/red,
        v44/v53/blue,v64/v72/blue,v84/v94/blue, v113/v121/blue,
      v65/v73/\dg,v82/v93/\dg}{
        \draw[\k,very thick] (\i) -- (\j) ;
      }
       \foreach \i/\j/\k/\h in {v72/v84/blue/20, v112/v123/\dg/20}{
        \draw[\k,very thick] (\i) to [bend left=\h] (\j) ;
      }

      \foreach \i/\j/\k in {v23/v35/red, v53/v64/blue, v73/v82/\dg}{
        \draw[\k,decorate,decoration={snake,amplitude=.5mm},very thick] (\i) -- (\j) ;
      }
    \end{tikzpicture}
    \caption{Illustration of $X$ (for $t=6$) and the first (i.e., with $s=1$) $f(t)$ layers (from left to right) of the connected component $C$ of $G-X$ rooted at $v$.
      Vertices of $C$ with a neighbor in $X$ are filled.
      We depict the first three iterations: $Q_1, e_1$ (in red), $Q_2, e_2$ (in blue), $Q_3, e_3$ (in dark green).
      Not to clutter the figure, we do not represent all the edges, and we code the different labels: edges $e_i$ are squiggly, and subpaths $Q'_i$ are thicker.}
    \label{fig:paths-Pi-Qi}
  \end{figure}

Let $e_1, \ldots, e_k$ be the obtained collection of edges.
In principle, the while loop stops when one of the following two conditions holds:
\begin{enumerate}[$(i)$]
\item $L_s$ is disconnected from $L_{s+f(t)-1}$ in $C \setminus \{e_1, \ldots, e_k\}$, or
\item there is no edge $e_{k+1} \in E(Q_{k+1})$ such that $\dist(e_{k+1},v) \geq \dist(e_k,v)+2$ and $e_{k+1}$ has no endpoint in $N_G(X)$.
\end{enumerate}

\begin{claim}
  \label{cl:no-case2}
  If case $(ii)$ holds, then $k > t^2$.
\end{claim}
\begin{claimproof}
  Remark first that for any~$i$ we have $\dist(e_i,v) \leq s + i(6t^2 + 2) - 2$.
  Next, we claim that as long as $\dist(e_i,v) + 2 + 6t^2 \leq s + f(t) - 1$, and case~$(i)$ does not occur,
  the construction of $e_{i+1}$ cannot fail, since $Q_{i+1}$ exists,
  and there are at least $6t^2$ edges~$e$ in~$Q_{i+1}$ satisfying $\dist(e,v) \ge \dist(e_i,v)$.
  Thus if case~$(ii)$ occurs, it must be that $(k+1)(6t^2 + 2) - 2 > f(t) - 1$,
  hence $k > t^2$ since $f(t) = (t^2+1)(6t^2 + 2)$.
\end{claimproof}

\begin{claim}
  \label{cl:k-bounded}
  It holds that $k \leq t^2$.
\end{claim}
\begin{claimproof}
Assume for the sake of contradiction that $k \geq t^2+1$.
By~\cref{lem:edgesinC} (using both items), for each~$i \in [k]$ there exists a~vertex $v_i \in V(C) \cap N_G(X)$ disconnected from $v$ in $C \setminus e_i$. 
Since $k \geq t^2+1$, there are $t+1$ vertices $v_{a_1}, \ldots, v_{a_{t+1}}$, with $a_1 < a_2 < \ldots < a_t < a_{t+1}$, all adjacent to an endpoint of the same fixed edge $e \in E(X)$.
Let $C'$ be the component containing $v$ in $C \setminus \{e_{a_1}, \ldots, e_{a_t}\}$, and let us contract (in $G$) the set $V(C')$ to a~single vertex, say $w$. We will now show that vertex $w$ is adjacent to an endpoint of~$e$.

By construction, the edge $e_{a_{t+1}}$ is reachable from $v$ in $C \setminus \{e_{a_1}, e_{a_2}, \ldots, e_{a_t}\}$, hence $e_{a_{t+1}} \in E(C')$.
Let $P'$ be a path in $C'$ from $v$ to an endpoint of $e_{a_{t+1}}$.
Every path in $C$ from $v$ to $v_{a_{t+1}}$ goes through at least one endpoint of $e_{a_{t+1}}$, by definition of $v_{a_{t+1}}$.
Let us consider a~shortest path in $C$ from $v$ to $v_{a_{t+1}}$, i.e., intersecting each layer at most once.
Consider a~suffix $P''$ of this path starting at an endpoint of $e_{a_{t+1}}$ and ending at $v_{a_{t+1}}$.
By the previous remark, $P''$ cannot contain an edge among $\{e_{a_1}, e_{a_2}, \ldots, e_{a_t}\}$, since these edges are in layers with strictly smaller indices than the endpoints of~$e_{a_{t+1}}$.
Thus $P' \cup P''$ (possibly combined with ${e_{a_{t+1}}}$) connects $v$ to $v_{a_{t+1}}$ in $C \setminus \{e_{a_1}, e_{a_2}, \ldots, e_{a_t}\}$.
Therefore, the contracted vertex $w$ \emph{contains} $v_{a_{t+1}}$; the latter being adjacent to an endpoint~$e$.

Observe also that, for every $i \in [t]$, each subpath $Q'_{a_i}$ is contained in $C'$.
Vertex $w$ is adjacent to an endpoint $u_{a_i}$ of $e_{a_i}$, for every $i \in [t]$.
Let $C_{a_i}$ be the vertex set of the connected component of $C \setminus e_{a_i}$ containing $v_{a_i}$.
Let $R_{a_i}$ be a~path from $u_{a_i}$ to $v_{a_i}$ in the subgraph of $G$ induced by $C_{a_i}$ plus the endpoints of $e_{a_i}$. 
For each $i \in [t]$, if $R_{a_i}$ has at least two edges, contract the edge $wu_{a_i}$, and all the edges of $R_{a_i}$ but the last two; the last one being $f_{a_i}=y_{a_i}v_{a_i}$.
Finally contract $e$, and the edge between $e$ and $w$.
We call the resulting vertex~$z$.

We claim that $z$ and the edges $f_{a_i}$ make a~$K_1+tK_2$ induced minor in $G$.
One can see that $z$ is adjacent to $v_{a_i}$ via $e \in E(X)$, and to $y_{a_i}$ via $w$ and the path $R_{a_i}$.
We shall justify that the edges $f_{a_i}$ form an induced matching in $G$.
Indeed, suppose some $f_{a_i}$ and $f_{a_j}$ touch with $i \neq j \in [t]$.
Then in $C \setminus e_{a_i}$, there is a path from $v$ to $v_{a_i}$ via $Q'_{a_j}$, a~contradiction.
\end{claimproof}

By \cref{cl:no-case2,cl:k-bounded}, case~$(ii)$ is impossible,
and the $2k \leq 2t^2$ endpoints of $e_1, \ldots, e_k$ form a~vertex cutset disconnecting $L_s$ from $L_{s+f(t)-1}$.
We can now build the path-decomposition $\mathcal P$.
Recall that the BFS search from $v$ gives rise to $q$ layers, $L_1, \ldots, L_q$ (outside $\{v\}$).

For $j \in [\lfloor q/f(t) \rfloor]$, let $S_j$ be the vertex cutset of size at most $2t^2$ (and obtained as detailed above) disconnecting $L_{(j-1)f(t)+1}$ from $L_{jf(t)}$.
We denote by \emph{$L_{s \rightarrow s'}$} the set $\bigcup_{s \leq h \leq s'} L_h$.
\begin{itemize}
\item Let $B_1 \subseteq V(C)$ consist of $v$, $S_1$, plus all the vertices of connected components of $C[L_{1 \rightarrow f(t)}]-S_1$ that do not intersect $L_{f(t)}$.
\item For $j$ going from 2 to $\lfloor q/f(t) \rfloor - 1$,
  let $B_j \subseteq V(C)$ consist of $S_j \cup S_{j+1}$ plus the vertices not already present in one of $B_1, \ldots, B_{j-1}$ of all the connected components of $$C[L_{(j-1)f(t)+1 \rightarrow (j+1)f(t)}]-(S_j \cup S_{j+1})$$ that do not intersect $L_{(j+1)f(t)}$.
\item Let finally $B_{\lfloor q/f(t) \rfloor} \subseteq V(C)$ consist of $S_{\lfloor q/f(t) \rfloor}$ plus the vertices of all the connected components of $C[L_{(\lfloor q/f(t) \rfloor-1)f(t)+1  \rightarrow q}]-S_{\lfloor q/f(t) \rfloor}$ that intersect $L_q$.
\end{itemize}
Let $\mathcal P$ be the path-decomposition $(B_1,B_2,\ldots,B_{\lfloor q/f(t) \rfloor})$.
$\mathcal P$ is indeed a path-decomposition of~$C$, since our process entirely covers $V(C)$, and by virtue of $S_j$ separating $L_{(j-1)f(t)+1}$ from $L_{jf(t)}$.
By construction,
\begin{itemize}
\item every bag intersects at most $2f(t)=O(t^4)$ layers of BFS from vertex~$v$, and
\item for every $j \in [\lfloor q/f(t) \rfloor - 1]$, $B_j \cap B_{j+1}=S_j$, so every adhesion has size at most $2t^2$.
\end{itemize}

Finally note that once $X$ is found, one can find the path-decomposition $\mathcal P$ of $C$ in time~$n^{O(1)}$, since this only involves computing (at most $n$) shortest paths.
\end{proof}

We can now wrap up, using~\cref{prop:path-dec,lem:bounded-diameter}.
\begin{reptheorem}{thm1}\label{thm:K1+tK2}
  For every positive integer $t$, \mis can be solved in polynomial-time $n^{O(t^5)}$ in $n$-vertex $K_1+tK_2$-induced-minor-free graphs.
\end{reptheorem}
\begin{proof}
  Let $G$ be our~$K_1 + tK_2$-induced-minor-free $n$-vertex input graph.
  As including vertices of degree 1 or vertices of degree 2 with adjacent neighbors in the independent set is a safe reduction rule (since it is more generally safe to include a~\emph{simplicial} vertex, i.e., one whose neighborhood is a~clique), we can assume that $G$ is reduced.
  By dealing with the possibly several connected components of $G$ separately, we can further assume that $G$ is connected.
  If $G$ has no $tK_2$ as an induced subgraph, we conclude by invoking Alekseev's result~\cite{Alekseev07,Tsukiyama77}.
  Thus we also assume that $G$ has such an induced subgraph.
  In time $n^{O(t)}$ we find $X \subseteq V(G)$ that maximizes the order of a~largest component of $G' := G - X$, among those sets $X$ such that $G[X] \simeq tK_2$.

  We exhaustively guess the intersection $X'$ of a fixed maximum independent set of $G$ with the set $X$, with an extra multiplicative factor of $2^{2t}$.
  We are now left with solving \smis separately in $C' := C-N_G(X')$ for each connected component $C$ of $G'$.
  By~\cref{prop:path-dec}, we obtain in time $n^{O(1)}$ a~path-decomposition $\mathcal P=(B_1,\ldots,B_p)$ of $C'$, such that
  \begin{itemize}
   \item every $B_i$ (for $i \in [p]$) is contained in $O(t^4)$ consecutive BFS layers \emph{of $C$}, and
   \item every adhesion $A_i := B_i \cap B_{i+1}$ (for $i \in [p-1]$) is of size at most $2t^2$.
  \end{itemize}
  Indeed, removing $N_G(X')$ from $C$ (and from its path-decomposition) preserves those properties.
  
  Let us define $A_0, A_p$ to be empty.
  We proceed to the following dynamic programming.
  For $i \in [0,p]$, and for any $S \subseteq A_i$, $T[i,S]$ is meant to eventually contain an independent set $I$ of $C'[\bigcup_{1 \leq j \leq i} B_j]$ of maximum cardinality among those such that $I \cap B_i=S$.
  We set $T[0,\emptyset]=\emptyset$, and observe that it is the only entry of the form $T[0,\cdot]$.

  We fill this table by increasing value of $i=1, 2, \ldots, p$.
  Assume that all entries of the form $T[i',\cdot]$ are properly filled for $i' < i$.
  For every $S \subseteq A_i$, $T[i,S]$ is filled in the following way.
  For every $S' \subseteq A_{i-1}$, if $S \cup S'$ is an independent set, we compute, by~\cref{lem:bounded-diameter} (with $L_1, \ldots, L_h$ being the $O(t^4)$ consecutive BFS layers \emph{of $C$} containing $B_i$, and $L_0$ being the connected set, \emph{in $C$}, formed by the union of all the previous layers), a~maximum independent set $I_i$ in $C'[B_i]-N[S \cup S']$ in time $n^{O(ht)}=n^{O(t^5)}$.
  We finally set $T[i,S]=T[i-1,S'] \cup I_i \cup S$ for a~run that maximizes the cardinality of $T[i-1,S'] \cup I_i$.

  It takes time $p \cdot 2^{O(t^2)} \cdot n^{O(t^5)}=n^{O(t^5)}$ to completely fill $T$.
  Eventually $T[p,\emptyset]$ contains a~maximum independent set of $C'$.
  We return the union of $X'$ and of the maximum independent sets of $C'$ found for each connected component $C$ of $G'$.
  The overall running time is $n^{O(t^5)}$.
\end{proof}

\section{Quasipolynomial algorithm in $tC_3 \uplus C_4$-induced-minor-free graphs}\label{sec:qp}

We first show that the existence of a~short hole allows for a~quasipolynomial-time branching rule in $tC_3 \uplus C_4$-induced-minor-free graphs.
By~branching rule, we mean here a~Turing reduction to (quasipolynomially many) subinstances with no short holes.

\begin{lemma}\label{lem:if-short-hole}
  Let $G$ be an $n$-vertex $tC_3 \uplus C_4$-induced-minor-free graph, and let $\ell \geqslant 4$ be a~fixed integer.
  While there is a~hole $H$ of length at most $\ell$ and a~$tC_3$ as an induced subgraph in $G$, \mis admits a quasipolynomial branching rule, running in $n^{O_\ell(t \log n)}$.
\end{lemma}
\begin{proof}
  Let $C(G) := \{X \in {V(G) \choose 3t}~:~G[X] \simeq tC_3\}$ be the collection of vertex subsets inducing $t$~disjoint triangles, and assume $\mu(G) := |C(G)| > 0$, and $H$ is a hole of $G$ of length at most~$\ell$.
  As $G$ is $tC_3 \uplus C_4$-induced-minor-free, $N[V(H)]$ intersects every $X \in C(G)$.
  In particular, there is a vertex $v \in V(H)$ such that $N[v]$ intersects at least a~$1/\ell$ fraction of the $X \in C(G)$.

  We branch on two options: either we take $v$ in an (initially empty) solution, and remove its closed neighborhood from $G$, or we remove $v$ from $G$ (without adding it to the solution).
  With the former choice, the measure $\mu$ drops by at least a~$1/\ell$ fraction (and the number of vertices of $G$ decreases by at least~1), and with the latter choice, the number of vertices drops by~1.
  This branching is exhaustive.
  We simply need to argue about its running time.

  Note that each option can be done at most $n$ times, while the first option cannot be done more than $\log_\ell(n^{3t})=O_\ell(t \log n)$ times.
  Hence the branching tree has at most ${n \choose \log_\ell(n^{3t})}=n^{O_\ell(t \log n)}$ leaves.
\end{proof}

The previous lemma permits us to get rid of short holes, which turns out useful in some corner case.

\begin{reptheorem}{thm2}
  For every positive integer $t$, \mis can be solved in quasipolynomial-time $n^{O(t^2 \log n)+f(t)}$ (where $f$ is single-exponential) in $tC_3 \uplus C_4$-induced-minor-free graphs.
\end{reptheorem}
\begin{proof}
  We apply the quasipolynomial branching rule of~\cref{lem:if-short-hole} with $\ell=6$, until the input $n$-vertex graph $G$ no longer has holes of length at most~6, or $tC_3$ induced subgraph.

  In time $n^{O(t)}$, we exhaustively look for a~collection of pairwise vertex-disjoint and non-adjacent triangles $T_1, T_2, \ldots, T_{t+2}$ in $G$.
  If such a collection does not exist, $G$ is $(t+2)C_3$-induced-minor-free.
  Indeed, the absence of $(t+2)C_3$ as an induced subgraph implies that at~least one of the $(t+2)$ independent cycles realizing a~$(t+2)C_3$ induced minor is of length at~least four.
  This is ruled out by the assumption that $G$ is $tC_3 \uplus C_4$-induced-minor-free.
  (Note here that only $(t+1)$ independent cycles would suffice.)

  We can thus assume that a collection $T_1, T_2, \ldots, T_{t+2}$ exists, for otherwise, we can conclude with the quasipolynomial-time algorithm, running in $n^{O(t^2 \log n)+f(t)}$ (where $f$ is single-exponential), of Bonamy et al.~\cite{Bonamy23} for \mis in graphs with a bounded number of independent cycles (here, $(t+2)C_3$-induced-minor-free).
  In turn, as $G$ contains $tC_3$ (even $(t+2)C_3$) as an induced subgraph, we can, in light of the first paragraph, further assume that all the cycles of $G$ have length either~3 or at least 7.
 We refer the reader to~\cref{fig:tC3cupC4} for a visual summary of the next two paragraphs.
  
  For every pair $T_i, T_j$ (with $i < j \in [t+2]$),
  consider the subgraph $G_{i,j} := G-\bigcup_{k \in [t+2] \setminus \{i,j\}} N(T_k)$.
  We claim that~$G_{i,j}$ is chordal.
  Indeed, since $G_{i,j}$ is disjoint from the neighborhood of $\bigcup_{k \in [t+2] \setminus \{i,j\}} T_k$,
  a hole~$H$ in~$G$ would form a $tC_3 \uplus C_4$-induced-minor together with $\{T_k~:~k \in [t+2] \setminus \{i,j\}\}$.
  Let now $S_{i,j}$ be a minimal $(T_i,T_j)$-separator in $G_{i,j}$.
  A classical argument then shows that~$S_{i,j}$ is a clique:
  suppose for the sake of contradiction that $u, v \in S_{i,j}$ are distinct and non-adjacent.
  Let~$X_i,X_j$ be the two components of~$G_{i,j} - S_{i,j}$ containing~$T_i,T_j$.
  By minimality of $S_{i,j}$, each of~$u,v$ is adjacent to both~$X_i$ and~$X_j$.
  Thus we can find two induced $uv$-paths, whose internal vertices are in~$X_i$ and~$X_j$ respectively.
  The union of these paths is a hole~$H$ in~$G_{i,j}$.

  Let $N_{i,j}$ be the set $N(T_i) \cap N(T_j)$ for each pair $i < j \in [t+2]$.
  Observe that the sets $N_{i,j}$ need not be disjoint, and that when $T_i$ and $T_j$ are at distance at least 3 apart, $N_{i,j}$ is empty.
  We notice that $N_{i,j}$ is a clique, for otherwise we can exhibit an induced cycle of length 4, 5, or 6 in $G$ (hence a hole of length at most~6).

  \begin{figure}[ht]
    \centering
    \begin{tikzpicture}[vertex/.style={draw,circle,inner sep=-0.08cm}]
      \foreach \i/\j/\k in {0/0/a,1/0/b,0.5/0.866/c, 4/0/d,5/0/e,4.5/0.866/f, 2/-2/g,3/-2/h,2.5/-1.134/i, 2.5/0/x1,2.5/0.4/x2,2.5/-0.4/x3,
      1.5/1/v1,2.5/1/v2,3.5/1/v3}{
        \node[vertex] (\k) at (\i,\j) {} ;
      }
      
      \foreach \i/\j/\k in {0.5/0.3/{$T_i$}, 4.53/0.3/{$T_j$}, 2.53/-1.7/{$T_k$}, 2.95/-0.8/{$N_{i,j}$}, 2.5/1.5/{$S_{i,j}$}, 2.5/0.4/{\footnotesize{$u$}}, 2.5/1/{\footnotesize{$v$}}}{
        \node at (\i,\j) {\k} ;
      }

      \foreach \i/\j/\k in {a/b/0,b/c/0,a/c/0, d/e/0,e/f/0,d/f/0, g/h/0,h/i/0,g/i/0, 
        b/x1/0,x1/d/0,x1/f/0, b/x2/0,x2/f/0, a/x3/-5,b/x3/0,x3/e/-5,x3/f/0,x3/d/0,x3/i/0,x1/i/-55, x1/x2/0,x1/x3/0,x2/x3/-30,
        c/v1/0,v1/v2/0,v2/v3/0,v3/f/0, v2/x2/0,v1/b/0,v1/x2/0,v3/x2/0}{
        \draw (\i) to [bend left=\k] (\j) ;
      }
      \foreach \i/\j/\k in {v1/x2/blue,v1/v2/blue, v2/v3/red,v3/x2/red,
        a/b/black,b/c/black,a/c/black, d/e/black,e/f/black,d/f/black, g/h/black,h/i/black,g/i/black}{
        \draw[\k,very thick] (\i) -- (\j) ;
      }

      \foreach \i/\j/\k in {{(x1)(x2)(x3)}/Nij/0.08, {(x2)(v2)}/Sij/0.12}{
        \node[draw,rectangle,thick,rounded corners,fit=\i,inner sep=\k cm,fill,fill opacity=0.1] (\j) {} ;
      }
    \end{tikzpicture}
    \caption{Illustration of the sets $S_{i,j}$ and $N_{i,j}$, and the two $uv$-paths through the two components, in red and blue.
      If~$uv$ were a non-edge, these paths would form a hole in the non-neighborhood of the other triangles $T_k$, contradicting $tC_3 \uplus C_4$-minor-freeness.
      The absence of hole of length at most~6 implies that $N_{i,j}$ is also a clique.}
    \label{fig:tC3cupC4}
  \end{figure}
  
  We claim that $G' := G - (\bigcup_{i < j \in [t+2]} S_{i,j} \cup N_{i,j})$ is chordal.
  Indeed assume there is a hole $H'$ in~$G'$.
  The $tC_3 \uplus C_4$-induced-minor-freeness implies that $H'$ intersects at least two sets $N[T_i]$ and $N[T_j]$.
  Thus there exists a subpath~$P$ of~$H'$ whose endpoints are in two distinct $N(T_i)$ and $N(T_j)$.
  By choosing~$P$ minimal, we can furthermore assume that no internal vertex of $P$ lies in some $N[T_k]$ with $k \notin \{i,j\}$.
  Since $G'$ does not include any vertex of $\bigcup_{i' < j' \in [t+2]} N_{i',j'}$,
  the endpoints of~$P$ are not in some $N[T_k]$ with $k \notin \{i,j\}$ either.
  Therefore, the path $P$ contradicts that $S_{i,j}$ separates $T_i$ and $T_j$ in $G-\bigcup_{k \in [t+2] \setminus \{i,j\}} N(T_k)$.

  We can now describe the rest of the algorithm after the collection $T_1, T_2, \ldots, T_{t+2}$ is found.
  We greedily compute the minimal separators $S_{i,j}$.
  We exhaustively try every subset $S \subseteq \bigcup_{i < j \in [t+2]} S_{i,j} \cup N_{i,j}$ that is an independent set.
  Such a~set $S$ contains at most one vertex in each $S_{i,j}$ and each $N_{i,j}$, as we have established that each $S_{i,j}$ and each $N_{i,j}$ form a~clique.
  Hence there are $n^{O(t^2)}$ such sets~$S$.
  For each $S$, we compute a~maximum independent set~$I$ in the \emph{chordal} graph $G-(N[S] \cup \bigcup_{i < j \in [t+2]} S_{i,j} \cup N_{i,j})$ in linear time (see~\cite{Gavril74,Rose76}).
  We finally output the set $S \cup I$ maximizing $|S \cup I|$.
  Note that the overall running time is $n^{O(t^2 \log n)+f(t)}$.
\end{proof}

\section{Conclusion}

We provided a~polynomial-time algorithm for \textsc{Maximum Independent Set} on graphs excluding the friendship graph as an induced minor
and a quasipolynomial-time algorithm on graphs excluding a disjoint union of $t$ triangles and a~4-cycle as an induced minor.
As mentioned in the introduction, it is of interest to study which other graph classes excluding some fixed planar graph $H$ as an induced minor \smis can be solved in (quasi)polynomial time.
If $H$ is a~subgraph of the friendship graph or of the wheel, some of our methods developed for $K_1+tK_2$ might extend.

For disjoint unions of cycles, the first open case is when $H$ is $C_4 \uplus C_4$.
Our treatment for $H = tC_3 \uplus C_4$ does not (easily) extend to this case.
It is thus likely that new methods have to be found.
Obtaining a quasipolynomial-time algorithm when $H=tC_t$ for any integer $t$ is a~first challenging milestone in the study of \smis on $H$-induced-minor-free graphs.


\begin{thebibliography}{10}

\bibitem{Abraham14}
Kuruvilla~J. Abraham and Clara Diaz.
\newblock Identifying large sets of unrelated individuals and unrelated
  markers.
\newblock {\em Source code for biology and medicine}, 9(1):1--8, 2014.

\bibitem{Abrishami22}
Tara Abrishami, Maria Chudnovsky, Cemil Dibek, and Paweł Rzazewski.
\newblock Polynomial-time algorithm for maximum independent set in
  bounded-degree graphs with no long induced claws.
\newblock In Joseph~(Seffi) Naor and Niv Buchbinder, editors, {\em Proceedings
  of the 2022 {ACM-SIAM} Symposium on Discrete Algorithms, {SODA} 2022, Virtual
  Conference / Alexandria, VA, USA, January 9 - 12, 2022}, pages 1448--1470.
  {SIAM}, 2022.
\newblock \href {https://doi.org/10.1137/1.9781611977073.61}
  {\path{doi:10.1137/1.9781611977073.61}}.

\bibitem{Abrishami24}
Tara Abrishami, Maria Chudnovsky, Marcin Pilipczuk, and Pawel Rzazewski.
\newblock Max weight independent set in sparse graphs with no long claws.
\newblock In Olaf Beyersdorff, Mamadou~Moustapha Kant{\'{e}}, Orna Kupferman,
  and Daniel Lokshtanov, editors, {\em 41st International Symposium on
  Theoretical Aspects of Computer Science, {STACS} 2024, March 12-14, 2024,
  Clermont-Ferrand, France}, volume 289 of {\em LIPIcs}, pages 4:1--4:15.
  Schloss Dagstuhl - Leibniz-Zentrum f{\"{u}}r Informatik, 2024.
\newblock URL: \url{https://doi.org/10.4230/LIPIcs.STACS.2024.4}, \href
  {https://doi.org/10.4230/LIPICS.STACS.2024.4}
  {\path{doi:10.4230/LIPICS.STACS.2024.4}}.

\bibitem{Abrishami21}
Tara Abrishami, Maria Chudnovsky, Marcin Pilipczuk, Paweł Rzążewski, and
  Paul~D. Seymour.
\newblock Induced subgraphs of bounded treewidth and the container method.
\newblock In D{\'{a}}niel Marx, editor, {\em Proceedings of the 2021 {ACM-SIAM}
  Symposium on Discrete Algorithms, {SODA} 2021, Virtual Conference, January 10
  - 13, 2021}, pages 1948--1964. {SIAM}, 2021.
\newblock \href {https://doi.org/10.1137/1.9781611976465.116}
  {\path{doi:10.1137/1.9781611976465.116}}.

\bibitem{Ahn25}
Jungho Ahn, J.~Pascal Gollin, Tony Huynh, and O{-}joung Kwon.
\newblock A coarse erd{\H{o}}s-p{\'{o}}sa theorem.
\newblock In Yossi Azar and Debmalya Panigrahi, editors, {\em Proceedings of
  the 2025 Annual {ACM-SIAM} Symposium on Discrete Algorithms, {SODA} 2025, New
  Orleans, LA, USA, January 12-15, 2025}, pages 3363--3381. {SIAM}, 2025.
\newblock \href {https://doi.org/10.1137/1.9781611978322.109}
  {\path{doi:10.1137/1.9781611978322.109}}.

\bibitem{Ahn20}
Sungsoo Ahn, Younggyo Seo, and Jinwoo Shin.
\newblock Learning what to defer for maximum independent sets.
\newblock In {\em Proceedings of the 37th International Conference on Machine
  Learning, {ICML} 2020, 13-18 July 2020, Virtual Event}, volume 119 of {\em
  Proceedings of Machine Learning Research}, pages 134--144. {PMLR}, 2020.
\newblock URL: \url{http://proceedings.mlr.press/v119/ahn20a.html}.

\bibitem{Alekseev82}
Vladimir~E. Alekseev.
\newblock The effect of local constraints on the complexity of determination of
  the graph independence number.
\newblock {\em Combinatorial-algebraic methods in applied mathematics}, pages
  3--13, 1982.

\bibitem{Alekseev04}
Vladimir~E. Alekseev.
\newblock Polynomial algorithm for finding the largest independent sets in
  graphs without forks.
\newblock {\em Discrete Applied Mathematics}, 135(1-3):3--16, 2004.
\newblock \href {https://doi.org/10.1016/S0166-218X(02)00290-1}
  {\path{doi:10.1016/S0166-218X(02)00290-1}}.

\bibitem{Alekseev07}
Vladimir~E. Alekseev.
\newblock An upper bound for the number of maximal independent sets in a~graph.
\newblock {\em Discrete Mathematics and Applications}, 17(4):355--359, 2007.
\newblock URL: \url{https://doi.org/10.1515/dma.2007.030} [cited 2023-02-03],
  \href {https://doi.org/doi:10.1515/dma.2007.030}
  {\path{doi:doi:10.1515/dma.2007.030}}.

\bibitem{AlKa00}
Paola Alimonti and Viggo Kann.
\newblock Some {APX}-completeness results for cubic graphs.
\newblock {\em Theoretical Computer Science}, 237(1):123 -- 134, 2000.
\newblock \href {https://doi.org/https://doi.org/10.1016/S0304-3975(98)00158-3}
  {\path{doi:https://doi.org/10.1016/S0304-3975(98)00158-3}}.

\bibitem{Alkhouri22}
Ismail~R. Alkhouri, George~K. Atia, and Alvaro Velasquez.
\newblock A differentiable approach to the maximum independent set problem
  using dataless neural networks.
\newblock {\em Neural Networks}, 155:168--176, 2022.

\bibitem{Augustson70}
J.~Gary Augustson and Jack Minker.
\newblock An analysis of some graph theoretical cluster techniques.
\newblock {\em Journal of the ACM (JACM)}, 17(4):571--588, 1970.

\bibitem{Bonamy23}
Marthe Bonamy, {\'E}douard Bonnet, Hugues D{\'e}pr{\'e}s, Louis Esperet, Colin
  Geniet, Claire Hilaire, St{\'e}phan Thomass{\'e}, and Alexandra Wesolek.
\newblock Sparse graphs with bounded induced cycle packing number have
  logarithmic treewidth.
\newblock In {\em Proceedings of the 2023 Annual ACM-SIAM Symposium on Discrete
  Algorithms (SODA)}, pages 3006--3028. SIAM, 2023.

\bibitem{BonnetBCTW18}
{\'{E}}douard Bonnet, Nicolas Bousquet, Pierre Charbit, St{\'{e}}phan
  Thomass{\'{e}}, and R{\'{e}}mi Watrigant.
\newblock Parameterized complexity of independent set in {H}-free graphs.
\newblock In Christophe Paul and Michał Pilipczuk, editors, {\em 13th
  International Symposium on Parameterized and Exact Computation, {IPEC} 2018,
  August 20-24, 2018, Helsinki, Finland}, volume 115 of {\em LIPIcs}, pages
  17:1--17:13. Schloss Dagstuhl - Leibniz-Zentrum f{\"{u}}r Informatik, 2018.
\newblock \href {https://doi.org/10.4230/LIPIcs.IPEC.2018.17}
  {\path{doi:10.4230/LIPIcs.IPEC.2018.17}}.

\bibitem{BonnetBTW19}
{\'{E}}douard Bonnet, Nicolas Bousquet, St{\'{e}}phan Thomass{\'{e}}, and
  R{\'{e}}mi Watrigant.
\newblock When maximum stable set can be solved in {FPT} time.
\newblock In Pinyan Lu and Guochuan Zhang, editors, {\em 30th International
  Symposium on Algorithms and Computation, {ISAAC} 2019, December 8-11, 2019,
  Shanghai University of Finance and Economics, Shanghai, China}, volume 149 of
  {\em LIPIcs}, pages 49:1--49:22. Schloss Dagstuhl - Leibniz-Zentrum f{\"{u}}r
  Informatik, 2019.
\newblock \href {https://doi.org/10.4230/LIPIcs.ISAAC.2019.49}
  {\path{doi:10.4230/LIPIcs.ISAAC.2019.49}}.

\bibitem{Bouchitte01}
Vincent Bouchitt{\'{e}} and Ioan Todinca.
\newblock Treewidth and minimum fill-in: Grouping the minimal separators.
\newblock {\em {SIAM} J. Comput.}, 31(1):212--232, 2001.
\newblock \href {https://doi.org/10.1137/S0097539799359683}
  {\path{doi:10.1137/S0097539799359683}}.

\bibitem{BrandstadtM18a}
Andreas Brandst{\"{a}}dt and Raffaele Mosca.
\newblock Maximum weight independent set for {\(\mathscr{l}\)}claw-free graphs
  in polynomial time.
\newblock {\em Discrete Applied Mathematics}, 237:57--64, 2018.
\newblock \href {https://doi.org/10.1016/j.dam.2017.11.029}
  {\path{doi:10.1016/j.dam.2017.11.029}}.

\bibitem{Butenko03}
Sergiy Butenko.
\newblock {\em Maximum independent set and related problems, with
  applications}.
\newblock University of Florida, 2003.

\bibitem{Butenko02}
Sergiy Butenko, Panos~M. Pardalos, Ivan Sergienko, Vladimir Shylo, and Petro
  Stetsyuk.
\newblock Finding maximum independent sets in graphs arising from coding
  theory.
\newblock In Gary~B. Lamont, Hisham Haddad, George~A. Papadopoulos, and
  Brajendra Panda, editors, {\em Proceedings of the 2002 {ACM} Symposium on
  Applied Computing (SAC), March 10-14, 2002, Madrid, Spain}, pages 542--546.
  {ACM}, 2002.
\newblock \href {https://doi.org/10.1145/508791.508897}
  {\path{doi:10.1145/508791.508897}}.

\bibitem{Chen22}
Li~Chen, Rasmus Kyng, Yang~P. Liu, Richard Peng, Maximilian~Probst Gutenberg,
  and Sushant Sachdeva.
\newblock Maximum flow and minimum-cost flow in almost-linear time.
\newblock In {\em 63rd {IEEE} Annual Symposium on Foundations of Computer
  Science, {FOCS} 2022, Denver, CO, USA, October 31 - November 3, 2022}, pages
  612--623. {IEEE}, 2022.
\newblock \href {https://doi.org/10.1109/FOCS54457.2022.00064}
  {\path{doi:10.1109/FOCS54457.2022.00064}}.

\bibitem{Chudnovsky20}
Maria Chudnovsky, Marcin Pilipczuk, Michał Pilipczuk, and St{\'{e}}phan
  Thomass{\'{e}}.
\newblock On the maximum weight independent set problem in graphs without
  induced cycles of length at least five.
\newblock {\em {SIAM} J. Discret. Math.}, 34(2):1472--1483, 2020.
\newblock \href {https://doi.org/10.1137/19M1249473}
  {\path{doi:10.1137/19M1249473}}.

\bibitem{Courcelle90}
Bruno Courcelle.
\newblock The monadic second-order logic of graphs. {I.} {R}ecognizable sets of
  finite graphs.
\newblock {\em Inf. Comput.}, 85(1):12--75, 1990.
\newblock \href {https://doi.org/10.1016/0890-5401(90)90043-H}
  {\path{doi:10.1016/0890-5401(90)90043-H}}.

\bibitem{Dabrowski12}
Konrad~K. Dabrowski, Vadim~V. Lozin, Haiko M{\"{u}}ller, and Dieter Rautenbach.
\newblock Parameterized complexity of the weighted independent set problem
  beyond graphs of bounded clique number.
\newblock {\em J. Discrete Algorithms}, 14:207--213, 2012.
\newblock \href {https://doi.org/10.1016/j.jda.2011.12.012}
  {\path{doi:10.1016/j.jda.2011.12.012}}.

\bibitem{Dallard-2}
Clément Dallard, Martin Milanič, and Kenny Štorgel.
\newblock Treewidth versus clique number. {II.} {T}ree-independence number.
\newblock {\em Journal of Combinatorial Theory, Series B}, 164:404--442, 2024.
\newblock URL:
  \url{https://www.sciencedirect.com/science/article/pii/S0095895623000886},
  \href {https://doi.org/10.1016/j.jctb.2023.10.006}
  {\path{doi:10.1016/j.jctb.2023.10.006}}.

\bibitem{Dallard-3}
Clément Dallard, Martin Milanič, and Kenny Štorgel.
\newblock Treewidth versus clique number. {III.} {T}ree-independence number of
  graphs with a forbidden structure.
\newblock {\em Journal of Combinatorial Theory, Series B}, 167:338--391, 2024.
\newblock URL:
  \url{https://www.sciencedirect.com/science/article/pii/S0095895624000236},
  \href {https://doi.org/10.1016/j.jctb.2024.03.005}
  {\path{doi:10.1016/j.jctb.2024.03.005}}.

\bibitem{Downey95}
Rodney~G. Downey and Michael~R. Fellows.
\newblock Fixed-parameter tractability and completeness {II}: On completeness
  for {W[1]}.
\newblock {\em Theor. Comput. Sci.}, 141(1{\&}2):109--131, 1995.
\newblock \href {https://doi.org/10.1016/0304-3975(94)00097-3}
  {\path{doi:10.1016/0304-3975(94)00097-3}}.

\bibitem{Eddy21}
Duncan Eddy and Mykel~J. Kochenderfer.
\newblock A maximum independent set method for scheduling earth-observing
  satellite constellations.
\newblock {\em Journal of Spacecraft and Rockets}, 58(5):1416--1429, 2021.

\bibitem{Forsyth46}
Elaine Forsyth and Leo Katz.
\newblock A matrix approach to the analysis of sociometric data: preliminary
  report.
\newblock {\em Sociometry}, 9(4):340--347, 1946.

\bibitem{Gardiner00}
Eleanor~J. Gardiner, Peter Willett, and Peter~J. Artymiuk.
\newblock Graph-theoretic techniques for macromolecular docking.
\newblock {\em J. Chem. Inf. Comput. Sci.}, 40(2):273--279, 2000.
\newblock \href {https://doi.org/10.1021/ci990262o}
  {\path{doi:10.1021/ci990262o}}.

\bibitem{GJ79}
Michael~R. Garey and David~S. Johnson.
\newblock {\em Computers and Intractability: {A} Guide to the Theory of
  NP-Completeness}.
\newblock W. H. Freeman, 1979.

\bibitem{Garey76}
Michael~R. Garey, David~S. Johnson, and Larry~J. Stockmeyer.
\newblock Some simplified {NP}-complete graph problems.
\newblock {\em Theor. Comput. Sci.}, 1(3):237--267, 1976.
\newblock \href {https://doi.org/10.1016/0304-3975(76)90059-1}
  {\path{doi:10.1016/0304-3975(76)90059-1}}.

\bibitem{Gartland20}
Peter Gartland and Daniel Lokshtanov.
\newblock Independent set on \emph{P\({}_{k}\)}-free graphs in quasi-polynomial
  time.
\newblock In Sandy Irani, editor, {\em 61st {IEEE} Annual Symposium on
  Foundations of Computer Science, {FOCS} 2020, Durham, NC, USA, November
  16-19, 2020}, pages 613--624. {IEEE}, 2020.
\newblock \href {https://doi.org/10.1109/FOCS46700.2020.00063}
  {\path{doi:10.1109/FOCS46700.2020.00063}}.

\bibitem{Gartland24}
Peter Gartland, Daniel Lokshtanov, Tom{\'{a}}s Masar{\'{\i}}k, Marcin
  Pilipczuk, Michal Pilipczuk, and Pawel Rzazewski.
\newblock Maximum weight independent set in graphs with no long claws in
  quasi-polynomial time.
\newblock In Bojan Mohar, Igor Shinkar, and Ryan O'Donnell, editors, {\em
  Proceedings of the 56th Annual {ACM} Symposium on Theory of Computing, {STOC}
  2024, Vancouver, BC, Canada, June 24-28, 2024}, pages 683--691. {ACM}, 2024.
\newblock \href {https://doi.org/10.1145/3618260.3649791}
  {\path{doi:10.1145/3618260.3649791}}.

\bibitem{Gartland21}
Peter Gartland, Daniel Lokshtanov, Marcin Pilipczuk, Michał Pilipczuk, and
  Paweł Rzążewski.
\newblock Finding large induced sparse subgraphs in
  \emph{C\({}_{\mbox{{\textgreater}t}}\)} -free graphs in quasipolynomial time.
\newblock In Samir Khuller and Virginia~Vassilevska Williams, editors, {\em
  {STOC} '21: 53rd Annual {ACM} {SIGACT} Symposium on Theory of Computing,
  Virtual Event, Italy, June 21-25, 2021}, pages 330--341. {ACM}, 2021.
\newblock \href {https://doi.org/10.1145/3406325.3451034}
  {\path{doi:10.1145/3406325.3451034}}.

\bibitem{Gavril74}
Fǎnicǎ Gavril.
\newblock The intersection graphs of subtrees in trees are exactly the chordal
  graphs.
\newblock {\em Journal of Combinatorial Theory, Series B}, 16(1):47--56, 1974.

\bibitem{GrzesikKPP22}
Andrzej Grzesik, Tereza Klimosov{\'{a}}, Marcin Pilipczuk, and Michał
  Pilipczuk.
\newblock Polynomial-time algorithm for maximum weight independent set on
  \emph{P}\({}_{\mbox{6}}\)-free graphs.
\newblock {\em {ACM} Trans. Algorithms}, 18(1):4:1--4:57, 2022.
\newblock \href {https://doi.org/10.1145/3414473} {\path{doi:10.1145/3414473}}.

\bibitem{Hastad96}
Johan Håstad.
\newblock Clique is hard to approximate within $n^{1-\varepsilon}$.
\newblock In {\em Acta Mathematica}, pages 627--636, 1996.

\bibitem{ImpagliazzoP01}
Russell Impagliazzo and Ramamohan Paturi.
\newblock On the complexity of {$k$-SAT}.
\newblock {\em J. Comput. Syst. Sci.}, 62(2):367--375, 2001.
\newblock \href {https://doi.org/10.1006/jcss.2000.1727}
  {\path{doi:10.1006/jcss.2000.1727}}.

\bibitem{ImpagliazzoPZ01}
Russell Impagliazzo, Ramamohan Paturi, and Francis Zane.
\newblock Which problems have strongly exponential complexity?
\newblock {\em J. Comput. Syst. Sci.}, 63(4):512--530, 2001.
\newblock \href {https://doi.org/10.1006/jcss.2001.1774}
  {\path{doi:10.1006/jcss.2001.1774}}.

\bibitem{Konig31}
Den{\'e}s K\H{o}nig.
\newblock Gr{\'a}fok {\'e}s m{\'a}trixok.
\newblock {\em Matematikai {\'e}s Fizikai Lapok}, 38:116--119, 1931.

\bibitem{Korhonen22}
Tuukka Korhonen.
\newblock Grid induced minor theorem for graphs of small degree.
\newblock {\em J. Comb. Theory, Ser. {B}}, 160:206--214, 2023.
\newblock \href {https://doi.org/10.1016/j.jctb.2023.01.002}
  {\path{doi:10.1016/j.jctb.2023.01.002}}.

\bibitem{Lamm17}
Sebastian Lamm, Peter Sanders, Christian Schulz, Darren Strash, and Renato~F.
  Werneck.
\newblock Finding near-optimal independent sets at scale.
\newblock {\em J. Heuristics}, 23(4):207--229, 2017.
\newblock \href {https://doi.org/10.1007/s10732-017-9337-x}
  {\path{doi:10.1007/s10732-017-9337-x}}.

\bibitem{Lozin08}
Vadim~V. Lozin and Martin Milanič.
\newblock A polynomial algorithm to find an independent set of maximum weight
  in a fork-free graph.
\newblock {\em J. Discrete Algorithms}, 6(4):595--604, 2008.
\newblock \href {https://doi.org/10.1016/j.jda.2008.04.001}
  {\path{doi:10.1016/j.jda.2008.04.001}}.

\bibitem{PilipczukPR21}
Marcin Pilipczuk, Michał Pilipczuk, and Paweł Rzążewski.
\newblock Quasi-polynomial-time algorithm for independent set in
  \emph{P\({}_{t}\)}-free graphs via shrinking the space of induced paths.
\newblock In Hung~Viet Le and Valerie King, editors, {\em 4th Symposium on
  Simplicity in Algorithms, {SOSA} 2021, Virtual Conference, January 11-12,
  2021}, pages 204--209. {SIAM}, 2021.
\newblock \href {https://doi.org/10.1137/1.9781611976496.23}
  {\path{doi:10.1137/1.9781611976496.23}}.

\bibitem{Poljak74}
Svatopluk Poljak.
\newblock A note on stable sets and colorings of graphs.
\newblock {\em Commentationes Mathematicae Universitatis Carolinae},
  15(2):307--309, 1974.

\bibitem{Pontoizeau21}
Thomas Pontoizeau, Florian Sikora, Florian Yger, and Tristan Cazenave.
\newblock Neural maximum independent set.
\newblock In {\em Machine Learning and Principles and Practice of Knowledge
  Discovery in Databases - International Workshops of {ECML} {PKDD} 2021,
  Virtual Event, September 13-17, 2021, Proceedings, Part {I}}, volume 1524 of
  {\em Communications in Computer and Information Science}, pages 223--237.
  Springer, 2021.
\newblock \href {https://doi.org/10.1007/978-3-030-93736-2\_18}
  {\path{doi:10.1007/978-3-030-93736-2\_18}}.

\bibitem{Robertson86}
Neil Robertson and Paul~D. Seymour.
\newblock Graph minors. {V. Excluding a planar graph}.
\newblock {\em J. Comb. Theory, Ser. {B}}, 41(1):92--114, 1986.
\newblock \href {https://doi.org/10.1016/0095-8956(86)90030-4}
  {\path{doi:10.1016/0095-8956(86)90030-4}}.

\bibitem{Rose76}
Donald~J. Rose, Robert~E. Tarjan, and George~S. Lueker.
\newblock Algorithmic aspects of vertex elimination on graphs.
\newblock {\em {SIAM} J. Comput.}, 5(2):266--283, 1976.
\newblock \href {https://doi.org/10.1137/0205021} {\path{doi:10.1137/0205021}}.

\bibitem{Tsukiyama77}
Shuji Tsukiyama, Mikio Ide, Hiromu Ariyoshi, and Isao Shirakawa.
\newblock A new algorithm for generating all the maximal independent sets.
\newblock {\em {SIAM} J. Comput.}, 6(3):505--517, 1977.
\newblock \href {https://doi.org/10.1137/0206036} {\path{doi:10.1137/0206036}}.

\bibitem{Verweij99}
Bram Verweij and Karen Aardal.
\newblock An optimisation algorithm for maximum independent set with
  applications in map labelling.
\newblock In Jaroslav Nesetril, editor, {\em Algorithms - {ESA} '99, 7th Annual
  European Symposium, Prague, Czech Republic, July 16-18, 1999, Proceedings},
  volume 1643 of {\em Lecture Notes in Computer Science}, pages 426--437.
  Springer, 1999.
\newblock \href {https://doi.org/10.1007/3-540-48481-7\_37}
  {\path{doi:10.1007/3-540-48481-7\_37}}.

\bibitem{Xiao17}
Mingyu Xiao and Hiroshi Nagamochi.
\newblock Exact algorithms for maximum independent set.
\newblock {\em Inf. Comput.}, 255:126--146, 2017.
\newblock \href {https://doi.org/10.1016/j.ic.2017.06.001}
  {\path{doi:10.1016/j.ic.2017.06.001}}.

\bibitem{Zuckerman07}
David Zuckerman.
\newblock Linear degree extractors and the inapproximability of max clique and
  chromatic number.
\newblock {\em Theory of Computing}, 3(1):103--128, 2007.
\newblock \href {https://doi.org/10.4086/toc.2007.v003a006}
  {\path{doi:10.4086/toc.2007.v003a006}}.

\end{thebibliography}
\end{document}